\documentclass[12pt,twoside]{article} 
\usepackage{geometry}

\usepackage{amsmath,amsthm}

\usepackage{xcolor}
\usepackage{color}
\usepackage{graphicx}
\usepackage{verbatim}
\usepackage{amssymb}
\usepackage{epstopdf}

\date{} 
\begin{document}

\centerline{} 

\centerline{\Large{\bf Scheduling a Rescue}} 
\centerline{} 

\centerline{\bf {Nian LIU}} 

\centerline{} 

\centerline{Department of Mathematics and Statistics} 

\centerline{Central South University} 

\centerline{Changsha, China} 

\centerline{} 

\centerline{\bf {Myron HLYNKA}} 

\centerline{} 

\centerline{Department of Mathematics and Statistics} 

\centerline{University of Windsor} 

\centerline{Windsor, Ontario, Canada N9B 3P4} 

\newtheorem{Theorem}{\quad Theorem}[section] 
\newtheorem{Definition}[Theorem]{\quad Definition} 
\newtheorem{Corollary}[Theorem]{\quad Corollary} 
\newtheorem{Lemma}[Theorem]{\quad Lemma} 
\newtheorem{Example}[Theorem]{\quad Example} 

\centerline{}

\begin{abstract}Scheduling service order, in a very specific queueing/inventory model with perishable inventory, is considered. Different strategies  are discusses and results are applied to the tragic cave situation in Thailand in June and July of 2018. 
\end{abstract} 


\section{Introduction} 

The order in which customers/items are served/processed  is a common problem in queueing theory (see \cite{Gro}). 
See Down et al. (\cite{Dow}) for queues with abandonment, and Raviv and Leshem for queues with deadlines. 
For scheduling, see Conway et al. (\cite{Con}) for discussion of RANDOM (Random), SPT (shortest processing time), 
MWKR (most work remaining), LWKR (least work remaining) orders of service.  
Nahmias (\cite{Nah}) discussed perishable inventory including issues of out-dating. Survival analysis is a standard topic (\cite{Hos})
in actuarial modeling.

In June and July of 2018, a group of thirteen persons (twelve boys and their coach) on a soccer team were trapped in flooded caves in Thailand. Their rescue attracted world attention.
One statement that initially appeared in the media stated that rescuers chose to take the strongest boys first. This was argued to give the best chance of survival.
Later statements said that the weaker boys actually were removed first. In the end, all thirteen people were rescued, but tragically, one of the rescuers did not survive. 
In this paper, we consider different models and criteria for which removing the stronger persons first may or may not be the best strategy.
The models to be presented do not fit precisely into standard scheduling or queueing or perishable inventory or survival analysis type models, and the objective here is different than in other settings. 

\section{Details }

In the actual Thailand situation, the people were rescued in batches (4+4+5). For our analysis, we will change the setting so that we have $n$ items, and they are processed one at a time. The items are labeled $1,2,\dots,n$ in the order of processing that is chosen (which can be changed as one chooses). Assume that at time 0, the $n$ items have probabilities of success $P_0(i), i=1,2,\dots,n$. (For the Thailand situation, we have $n=13$). 
Let $P_0=(P_0(1),\dots,P_0(n))$. 

Two possible measures of success are (a) the expected number of successfully processed items (expected number of rescued people) (b) the probability that all items are successfully processed (probability that all people are rescued).

First we assume that the initial probabilities of success do not change over time. 
Let $X_i=1$  if item $i$ is successfully processed and $X_i=0$ if item $i$ is not successfully processed.  Let $Y=X_1+\dots +X_{n}$ be the total number of successes among the $n$ items.  
Thus we are dealing with a generalization of the binomial distribution with differing $p$'s. This topic has been studied extensively (see \cite{Bol}, \cite{Dar}, \cite{Hoe}, \cite{Sam}).  
For a generalized binomial distribution with $n$ independent trials with success probabilities $(p_1,\dots,p_n)$, let $X_i=$ the success probability of item $i$, and let $Y=$ total number of successes. Then the probability of exactly $k$ successes in the $n$ trials is the coefficient of $z^k$ in the expansion of $\prod_{i=1}^n  (1-p_i+p_i z)$. 

\begin{Theorem}
For the generalized binomial model, both $E(Y)$ and \\
$P(\text{all successes})$ are constant regardless of the order in which items occur. 
\end{Theorem}
\begin{proof}
Note that $E(X_i)=0p_i+1(1-p_i)=p_i$. So \\
\begin{equation*}
E(Y)=\sum_{i=1}^n E(X_i)=\sum_{i=1}^n p_i \text{ and } P(\text{all successes})=\prod_{i=1}^n p_i.
\end{equation*}
Since the expressions for $E(Y)$ and $P(\text{all successes})$ do not change when we change the order of the $p_i$, the result follows. 
\end{proof}
In the context of the rescue problem, the order of rescuing the 13 people does not matter as far as the expected number of successes or the probability of complete success is concerned, if the probabilities of success do not change over time. 

It is more reasonable, however, to assume that as time goes by, the probabilities change  from the original vector $P_0$. Let  $0=t_1<t_2<\dots<t_n$ be the times of start of processing for the $i$-th ($i=1,2,\dots, n$) item. Then $T_i=t_{i+1}-t_i$ ($i=1,\dots,n$) is the interval between individual processing time starts. 
The probability of success of processing the $i$-th item is originally $P_0(i)$ but becomes $P_1(t_i)$ (to be specified) at the time that the $i$-th item is processed because of the delay to begin processing. Let $P_1=(P_1(t_1), P_1(t_2),\dots, P_1(t_n))$. So
\begin{align*}
&E(\text{number of successes})=\sum_{i=1}^n E(X_i)=\sum_{i=1}^n P_1(t_i) \\
&P(\text{all successes})=\prod_{i=1}^n P_1(t_i).
\end{align*}

\subsection{Additive Model}
Assume that the interval between the processing start time of consecutive items is a constant value $T$. Then the ordered processing start times for the $n$ items are $(0, T, 2T, \dots, (n-1)T)=(t_1,t_2,\dots,t_{n})$ so $t_i=(i-1)T$. We assume an additive model for the values $P_1(t_i)$. Use the notation $x^+=max(x,0)$. Specifically, we assume $P_1(t_i)=\max\{P_0(i)+f(t_i),0\} \equiv (P_0(i)+f(t_i))^+$, where the function $f$ is chosen to take negative values ($f(t)<0$), and we assume $f'(t)\leq 0$. The rationale for this assumption is that the probability of success decreases for every item over time.  

\begin{Theorem}
If $P_0(i)+f(t_i) >0$, $\forall i$, then 
\begin{align}
E(\# successes)&= \sum_{i=1}^{n} P_0(i) +\sum_{i=1}^{n} f((i-1)T) \\
P(\text{all successes}) &= \prod_{i=1}^{n} (P_0(i)+f((i-1)T))
\end{align}
\end{Theorem}
\begin{proof} 
\begin{align}
E(\# successes) &= \sum_{i=1}^{n} P_1(t_i) \notag \\
&= \sum_{i=1}^{n} (P_0(i)+f(t_i))^+= \sum_{i=1}^{n} (P_0(i)+f(t_i)) \notag\\
&= \sum_{i=1}^{n} P_0(i) +\sum_{i=1}^{n} f((i-1)T) \notag\\
P(\text{all successes}) &= \prod_{i=1}^{n} P(\text{i-th item is success})= \prod_{i=1}^{n} P_1(i)\notag\\
&= \prod_{i=1}^{n} (P_0(i)+f((i-1)T))\notag
\end{align}
\end{proof}
\begin{Corollary} The expected  number of successes is independent of the order of service. 
\end{Corollary}
\begin{proof}
\begin{equation*}
E(\# successes) = \sum_{i=1}^{n} P_0(i) +\sum_{i=1}^n f((i-1)T). 
\end{equation*}
This expression does not depend on the order of service. \end{proof}

Thus, in terms of expected number of successes, the order of the service does not matter. 
However the expression $P(\text{all successes})$ does depend on the order of service. Recall the arithmetic/geometric mean inequality. 
``If numbers $x_1,\dots,x_n$ are not all the same, the geometric mean of these numbers is less than their arithmetic mean.''(\cite{Kor}) Equality occurs only when the $x_i$ are all equal. So, with $\sum_{i=1}^{n} P_0(i)+f((i-1)T)$ fixed, the product of $n$ terms would be largest when terms are closest to each other.
In terms of the Thailand rescue, the probability that all 13 people survive would be the maximized if the weakest person among the remaining is saved first. But the expected number of people successfully rescued stays the same regardless of the order chosen.  Thus, we would likely choose to rescue the weaker people first in order to satisfy objective (b). 

As a simple example with $n=4$, take $P_0=(.8,.9,.7,.7)$  and $f(x)=-.1x/T$ so $\{f(t_i)\}=\{f((i-1)T)\}=\{-.1(i-1)\}=\{0,-.1,-.2,-.3\}$. Thus $P_1=(P_1(t_1), P_1(t_2),P_1(t_3),P_1(t_4))=(.8-0,.9-.1,.7-.2,.7-.3)=(.8,.8,.5,.4)$.\\
So  $E(\# successes)=  \sum_{i=1}^{n} P_1(t_i)=.8+.8+.5+.4 = 2.5$ and \\
$P(\text{all successes}) =  \prod_{i=1}^{n} P_1(t_i)=.8(.8)(.5)(.4) = .128$. \\
If we change the order  of $P_0$ to (.7,.7,.8,.9) but leave $f$ in its current form, we obtain $E(\# successes) = 2.5$ and $P(\text{all successes})=.1512$.\\
If we change the order of $P_0$ to (.9,.8,.7,.7) but leave $f$ in its current form, we obtain    $E(\# successes) = 2.5$ and $P(\text{all successes})=.126$. We see that using an increasing order in $P_0$ gives the optimal (highest) product. 

Next we consider the case when there exists an ordering such that for some $i$, we have $P_0(i)+f(t_i) < 0$.
\begin{Theorem} If for some $i$ we have $P_0(i)+f(t_i) < 0$, then
\begin{align}
E(\# successes) &= \sum_{i=1}^{n} P_1(t_i)= \sum_{i=1}^{n} (P_0(i)+f(t_i))^+ \\
P(all successes) &=  \prod_{i=1}^{n} P_1(t_i)= \prod_{i=1}^{n} (P_0(i)+f(t_i))^+ = 0
\end{align}
\end{Theorem}

We change our earlier example and choose $P_0=(.8,.9,.1,.2)$, but leave $f$ as before with $\{f(t_i)\}=\{f((i-1)T)\}=\{-.1(i-1)\}=\{0,-.1,-.2,-.3\}$.
Then  $\{P_1(i)\}=(.8-0,.9-.1,.1-.2,.2-.3)=(.8,.8,(-.1)^+,(-.1)^+.4)= (.8,.8,0,0)$\\
So  $E(\# successes)=  \sum_{i=1}^{n} P_1(t_i)=.8+.8+0+0 = 1.6$ and \\
$P(\text{all successes}) =  \prod_{i=1}^{n} P_1(t_i)=.8(.8)(0)(0) = 0$. \\
If we change the order  of $P_0$ to (.1,.2,.8,.9) but leave $f$ in its current form, we obtain $E(\# successes) = .1+.1+.6+.6=1.4$ and $P(\text{all successes})=.1(.1)(.6)(.6)= .0036$.\\
If we change the order of $P_0$ to (.9,.8,.2,.1) but leave $f$ in its current form, we obtain  $E(\# successes) = .9+.7+0+0=1.6$ and $P(\text{all successes})=.9(.7)(0)(0)= 0$.\\
We see that we maximize $E(\# successes)$ by choosing $P_0$ to have values in decreasing order but we maximize $P(\text{all successes})$ by choosing $P_0$ components to be in increasing order. 
In terms of the Thailand rescue, if for some $i$ we have $P_0(i)+f(t_i) < 0$, then we achieve a higher expected number of people rescued, if we rescue the stronger people first. This requires us to have a good estimate of the initial probabilities. However, we have a larger value of $P(\text{all successes})$ if we rescue the weaker people first. 

We can illustrate this as follows. We create a $13\times 13$ matrix. Each column represents a person. The rows represent time and the $(i,j)$ entry represents the probability of success if person $j$ is rescued on the $i$th step (at time $t_i$). For row 1, we generate 13 random values  uniformly on (.5,1) to represent the probabilities $P_0(j)$of success  for person $j$ at time 0. The people are labeled so that each row is in increasing order, with the weaker person (lowest success probability) is in column 1 and the strongest person is in column 13. To get the later rows, we use the function $f$, described earlier. This time, we choose $f(x)=-0.06x/T$. Then the $(2,j)$ element in row 2 is given by \\
$(P_0(j)-.06(2-1))^+$, and the (3,j) element of row 3  consists of entry\\ $(P_0(j)-.06(3-1))^+$, and so on.   
The probabilities are presented in a matrix shown in Figure \ref{color_matrix}, where the square at $i$-th row and $j$-th column corresponds to the $j$-th person's probability of survival if he is saved at the $i$-th stage. The bigger the square is and the darker its color, the bigger the probability.

\begin{figure}[htb]
\centering
\includegraphics[scale=0.5]{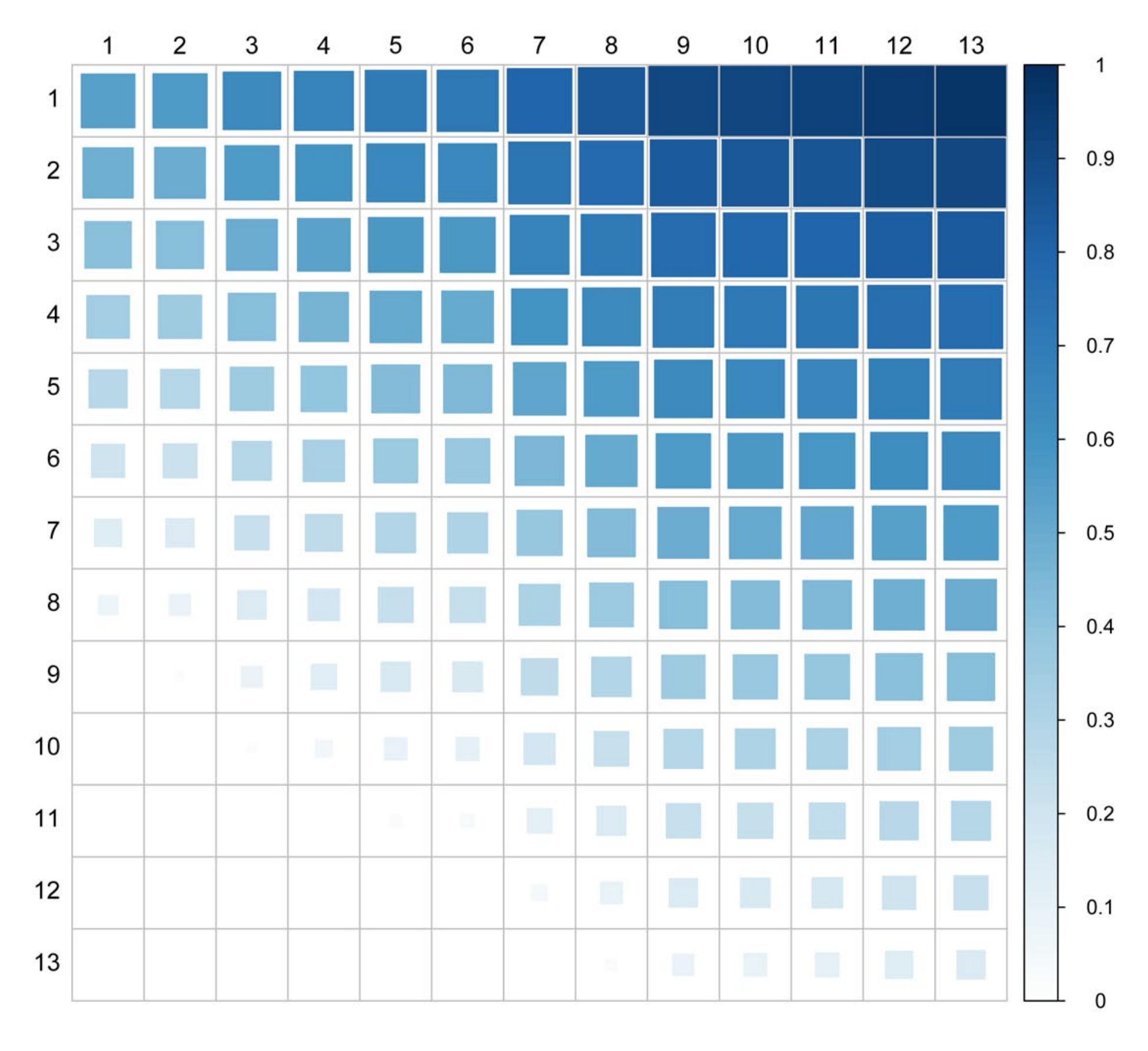}
\caption{Prob that customer $j$ survives if served at stage $i$}
\label{color_matrix}
\end{figure}

See Figure \ref{color_matrix}. We must choose one entry from each row and each column to indicate the order that the rescue attempts are made. If we always choose the weakest person to be taken earlier, then we are using the diagonal of the matrix from upper left to lower right. 
If we always choose the strongest person to be taken earlier, then we are using the diagonal from lower left to upper right. The product (and hence $P(\text{all survive})$) is maximized if the colors are closer to each other. It seems clear that the product of all entries for the case of stronger first strategy will give a value nearer to zero than the weaker first strategy because of the almost white values in the lower
left corner. So we maximize $P(\text{all survive})$ by rescuing the weaker people first. However, it is not so clear as to the optimal order of rescue in terms of the expected values since we must sum the values rather than multiply them. 

\section {Evaluating the model}

Our matrix in Figure \ref{color_matrix} was obtained by choosing the probabilities of success i.i.d. from $Unif(.5,1)$ and that these probabilities drop by $.06$  to a minimum of 0
(after each of the first 12 stages, so $f(x)=-.06x/T$. For such a model, define \\
$\pi_s = P(\text{all successes with strongest first}) $ and \\$\pi_w = P(\text{all successes with weakest first})$.\\
 It is possible that the values that we generate by simulation will result in 
$\pi_s=0$ or $\pi_s>0$ (similarly for $\pi_w$). As an evaluation of the model, we next compute $P(\pi_s>0)$ and $P(\pi_w>0)$, if we were to simulate the values. 

First consider $P(\pi_s>0)$.  For our generator, let $P_0$ be the vector of initial probabilities generated. Let $P_1=(P_1(t_1),\dots,P_1(t_n))$ be the vector of success probabilities for the items $(1,\dots,n)$ at the processing start time for those items. Then $\pi_s=0$  iff $\exists k, s.t.\, P_1(t_k)=0$. As earlier, take $t_i=(i-1)T$.  
\begin{align}
P(\pi_s >0) &=P(\min \{P_1(t_k)\} >0) = P(\min \{P_0(k)\} +f(t_{13})>0)\notag\\
&= P(\min \{P_0(k)\} +f(12T))>0)= P(\min \{P_0(k)\} -.72>0) 
\notag \\
&=\prod\limits_{k=1}^{13}P(P_0(k) > .72)=(\frac{1-.72}{.5})^{13}=.000533 
\end{align}

If the weak ones were saved first, then the initial cases (all with $P_0(i)>.5$) would all have positive probabilities and the only cases with potentially zero probabilities generated would be
 the largest four values which must be larger than .06*.9, .06*10, .06*11,.6*12 respectively. The probability of such values being generated can be found frm the joint distribution of the largest four order statistics of values generated on Unif(.5,1). 
 If we call these order statistics $t,u,v,w$ (from smaller to larger), then the joint distribution of the largest four order statistics is $g(t,u,v,w) = \dfrac{13!}{9!1!1!1!1!}(2t-1)^9 2^4.$ 
 This would be computed as

\begin{align*}
 P(\pi_w >0) &= P(\bigcap_{k=1}^{13} \{ P_0(k) > -f(T(k-1))\}  \\
&=\int_{.72}^1\int_{.66}^w\int_{.6}^v\int_{.54}^u  g(t,u,v,w) dt\, du\, dv\, dw=.9999677
\end{align*}

Thus, for the model that we are using to generate our values, if we are most concerned with $P(\text{all successes})$, then it is best to processes highest probability of success items first (i.e. save the stronger boys first).

\section {Mutiplicative Model}

In the previous analysis, we assumed that we had an initial probability vector and that the probabilities decreased over time in an additive manner.
Howver, it might be more reasonable to assume that the decrease in probabilities over time is multiplicative rather than additive. 
i.e. We begin wth $P_0=(P_0(1),\dots,P_0(n)$ as the probabilities of success for $n$ items listed in order of processing. 
Let $p\in (0,1)$ be the multiplciative factor resulting in the updated probability. After the first item is processed, the probabilites of the remaining $n-1$ items become $(P_0(2)*p,\dots, P_0(n)*p)$, etc.
At the time of processing, the vector of probabilites of success is \\
$P_1=(P_1(1),P_1(2),\dots,P_1(n))=(P_0(1),P_0(2)*p,\dots,,P_0*p^n)$.

\begin{Theorem} For initial success vector  $P_0=(P_0(1),\dots,P_0(n))$and the multiplicative factor $p$, 
\begin{align}
E(\# successes)&=\sum_{i=1}^n P_0(i)p^i\\
P(\text{all successes})&=\prod_{i=1}^n P_0(i)p^{(n-1)n/2}.
\end{align}
\end{Theorem}
\begin{proof}
\begin{align*}
E(\# successes)&=\sum_{i=1}^n P_1(i)=\sum_{i=1}^n P_0(i)p^i\\
P(\text{all successes})&=\prod{i=1}^n P_1(i)=\prod_{i=1}^n P_0(i)P^i=\prod_{i=1}^n P_0(i)p^{(n-1)n/2}.
\end{align*}\end{proof}

\begin{Corollary}
In the multiplicative model, $P(\text{all successes})$ is independent of the order of processing.
\end{Corollary}

\begin{Corollary}
In the multiplicative model, $E(\# successes)$ is maximized if $P_0$ is sorted so that $P_0(1)\geq \dots \geq P_0(n)$
\end{Corollary}
\begin{proof} The result can be shown by induction. \end{proof}

In the case of the cave rescue, the multiplicative model indicates that the order of rescue does not affect
$P(\text{all successes}$, but that in order to maximize the number of people saved, the stronger people should be rescued first.

\section{Conclusion}
In a rescue situation, one needs to clarify what the goal is. Generally, the logical goal would be to maximize the expcted number of successes.
In the additive model, the order does not matter, unless some of the items have their probabilities drop too much by the delay, in which case 
the higher order items should be processed first. In the multiplicative model, the preferred order is to rescue the higher success probability items first.
This conflicts with our intuition, and it also contradicts our sense of fairness. A seemingly less important goal of maximizing the probability that all items are
successfully processed, results in the opposite preferred ordering.  

\noindent
{\bf Acknowledgements.} 
We acknowledge funding and support  from MITACS Global Internship program, University of Windsor, Central South University, CSC Scholarship.

\end{document}